\documentclass[submission,copyright,creativecommons]{eptcs}
\providecommand{\event}{GandALF 2026} % Name of the event you are submitting to

\usepackage{iftex}

\ifpdf
  \usepackage{underscore}         % Only needed if you use pdflatex.
  \usepackage[T1]{fontenc}        % Recommended with pdflatex
\else
  \usepackage{breakurl}           % Not needed if you use pdflatex only.
\fi

\usepackage{amssymb}
\usepackage{amsthm}
\usepackage{stmaryrd}
\usepackage{mathtools}
\usepackage{xspace}
\usepackage{mathpartir}
\usepackage{tikz}
\usepackage{hyperref}
\usepackage{thm-restate}

\title{Capturing Bisimulation-Invariant Exponential-Space Complexity Classes}
\author{Florian Bruse
\institute{TU Munich \\ School of Comp.\, Inform.\ and Techn. \\ Munich, Germany}
\and
David Kronenberger 
\hspace*{1.5cm} Martin Lange
\institute{University of Kassel\\ Theor.\ Comp.\ Sci. / Formal Methods \\ Kassel, Germany}
}
\def\titlerunning{Capturing Bisimulation-Invariant Exponential-Space Complexity Classes}
\def\authorrunning{F. Bruse, D. Kronenberger \& M. Lange}
\begin{document}
\maketitle

\theoremstyle{plain}
\newtheorem{theorem}{Theorem}
\newtheorem{lemma}[theorem]{Lemma}
\newtheorem{proposition}[theorem]{Proposition}

\theoremstyle{definition}
\newtheorem{definition}[theorem]{Definition}
\newtheorem{example}[theorem]{Example}

\renewcommand{\epsilon}{\varepsilon}

\newcommand{\M}[1]{\ensuremath{\mathcal{#1}}}
\newcommand{\Nat}{\ensuremath{\mathbb{N}}}
\newcommand{\trPHFL}{\text{trPHFL}\xspace}
\newcommand{\PHFL}{\text{PHFL}\xspace}
\newcommand{\HFL}{\text{HFL}\xspace}
\newcommand{\phfl}[2]{\PHFL^{#1}_{#2}}
\newcommand{\trphfl}[2]{\trPHFL^{#1}_{#2}}
\newcommand{\hfl}[1]{\HFL^{#1}}
\newcommand{\HOPFP}{\hopfp{}}
\newcommand{\HOLLFP}{\hollfp{}}
\newcommand{\hollfp}[1]{\ensuremath{\textsf{HO}^{#1}\textsf{(LFP)}}\xspace}
\newcommand{\hopfp}[1]{\ensuremath{\textsf{HO}^{#1}\textsf{(PFP)}}\xspace}
\newcommand{\FOPFP}{\ensuremath{\textsf{FO(PFP)}\xspace}}
\newcommand{\FOLFP}{\ensuremath{\textsf{FO(LFP)}\xspace}}
\newcommand{\mucalc}{\ensuremath{\mathcal{L}_\mu}\xspace}
\newcommand{\polymucalc}{\ensuremath{\mucalc^\omega}\xspace}

\newcommand{\bisim}[1]{\ensuremath{\text{#1}{\mathbin{/_{\sim}}}}}

\newcommand{\ExpTime}{\textsf{EXPTIME}\xspace}
\newcommand{\ExpSpace}{\textsf{EXPSPACE}\xspace}
\newcommand{\PTime}{\textsf{P}\xspace}
\newcommand{\PSpace}{\textsf{PSPACE}\xspace}
\newcommand{\NPTime}{\textsf{NP}\xspace}
\newcommand{\NLogSpace}{\textsf{NLOGSPACE}\xspace}

\newcommand{\Transsys}{\ensuremath{\mathcal{T}}}
\newcommand{\States}{\ensuremath{{S}}}
\newcommand{\Prop}{\ensuremath{\mathbf{P}}}
\newcommand{\Actions}{\ensuremath{\mathbf{A}}}
\newcommand{\Label}{\ensuremath{\ell}}
\newcommand{\state}{\ensuremath{s}}
\newcommand{\statealt}{\ensuremath{t}}

\def\newarrow#1{\mathop{{\hbox{\setbox0=\hbox{$\scriptstyle{#1\quad}$}{$%
\mathrel{\mathop{\setbox1=\hbox to
\wd0{\rightarrowfill}\ht1=3pt\dp1=-2pt\box1}\limits^{#1}}%
$}}}}}

\def\newarrowi#1#2{\mathop{{\hbox{\setbox0=\hbox{$\scriptstyle{#1\quad}$}{$%
\mathrel{\mathop{\setbox1=\hbox to
\wd0{\rightarrowfill}\ht1=3pt\dp1=-2pt\box1}\limits^{#1}}_{#2}%
$}}}}}

\newcommand{\Transition}[3]{\ensuremath{#1 \newarrow{#2} #3}}
\newcommand{\Transitioni}[4]{\ensuremath{#1 \newarrowi{#2}{#4} #3}}

\newcommand{\indtype}{\odot}
\newcommand{\ORDHO}{\mathit{order}}
\newcommand{\ordho}[1]{\ORDHFL(#1)}
\newcommand{\hovars}{\ensuremath{\mathbf{V}}}
\newcommand{\pfp}{\ensuremath{\mathit{pfp}}}
\newcommand{\PFP}{\ensuremath{\mathit{PFP}}}

\newcommand{\grtype}{\bullet}
\newcommand{\varpos}{\ensuremath{+}}
\newcommand{\varneg}{\ensuremath{-}}
\newcommand{\varboth}{\ensuremath{0}}

\newcommand{\dimension}{\ensuremath{d}}
\newcommand{\ORDHFL}{\mathit{ord}}
\newcommand{\ordhfl}[1]{\ORDHFL(#1)}

\newcommand{\sem}[3]{\ensuremath{\llbracket #1 \rrbracket^{#2}_{#3}}}
\newcommand{\semm}[2]{\sem{#1}{#2}{}}

\newcommand{\tail}{\vdash_{\textsc{tail}}}

\newcommand{\mytrue}{\ensuremath{\mathtt{t\!t}}}
\newcommand{\myfalse}{\ensuremath{\mathtt{f\!f}}}
\newcommand{\lvars}{\ensuremath{\mathbf{L}}}
\newcommand{\fvars}{\ensuremath{\mathbf{F}}}
\newcommand{\fvarsbis}{\ensuremath{\mathbf{G}}}
\newcommand{\mudiam}[2]{\ensuremath{\langle #1 \rangle_{#2}}}
\newcommand{\mubox}[2]{\ensuremath{[#1]_{#2}}}
\newcommand{\subst}[2]{\{#1\}#2}

\newcommand{\nxt}{\mathit{next}}

\newcommand{\SEMTRANS}{\mathit{tptr}}
\newcommand{\semtrans}[3]{\SEMTRANS^{#2}_{#3}(#1)}
\newcommand{\SYNTRANS}{\mathit{trans}}
\newcommand{\syntrans}[1]{\SYNTRANS(#1)}

\allowdisplaybreaks

\begin{abstract}
Otto's Theorem characterises the bisimulation-invariant PTIME queries over graphs as exactly
those that can be formulated in the polyadic $\mu$-calculus. 
This extends to a characterisation of the full bisimulation-invariant exponential-time hierarchy,
using a higher-order extension of the polyadic $\mu$-calculus, called PHFL, which is naturally
stratified by type order. The bisimulation-invariant fragment of $k$-EXPTIME is captured
by the fragment of PHFL whose formulas have type order at most $k$. This is obtained using a
generalisation of the Immerman-Vardi Theorem by Freire and Martins.

Here we provide such a characterisation of the exponential-space hierarchy. We show
that the bisimulation-invariant fragment of $k$-EXPSPACE is captured by the fragment of PHFL of type 
order $k+1$ whose formulas are tail-recursive. Tail recursion limits the interplay
between fixpoints and higher-order functions so that model checking based on fixpoint iteration
can be made more space efficient in a way comparable to tail recursion in programming languages. The capturing 
result makes use of a generalisation of Vardi's characterisation of PSPACE as first-order logic with 
partial fixpoints, namely that $k$-EXPSPACE is captured by higher-order logic of order $k+1$ with 
partial fixpoints.
\end{abstract}

\section{Introduction}

Descriptive complexity is the subfield of complexity theory that characterises computational problems by means of the
resources needed to \emph{describe} them, rather than the resources (like time and space in some model of computation)
needed to \emph{solve} them. It is appealing because it yields characterisations of complexity classes that are 
independent of an underlying machine model. Hence, it links computational power to expressive power. The primary 
aim in descriptive complexity is to show that a complexity class $\mathcal{C}$ is \emph{captured} by a logic $\mathcal{L}$,
meaning that (\textsc{i}) the model checking problem for $\mathcal{L}$ over finite structures is in class $\mathcal{C}$, 
and (\textsc{ii}) for every problem $P \in \mathcal{C}$ there is an $\mathcal{L}$-formula describing exactly the positive
instances of $P$.

The study of descriptive complexity started with Fagin's result \cite{Fagin74}, stating that the class \NPTime is captured 
by $\exists$SO, the existential fragment of second-order logic. With tools at hand to separate the expressiveness of 
logics like Ehrenfeucht-Fra\"{\i}ss{\'e} games \cite{Eh61,Fra-quelques-crb}, 
descriptive complexity opens up new ways to separate complexity classes. Unfortunately, for the famous case of
\NPTime vs.\ \PTime, no logical characterisation is known so far, despite ongoing effort in this direction, cf.\  
\cite{DBLP:conf/lics/Grohe08,DBLP:journals/siglog/DawarP24}.
The major obstacle here is the fact that structures are naturally ordered when represented as inputs to Turing machines
or other forms of algorithms, while they are naturally unordered in the eyes of a logical formula. For logics like
$\exists$SO and beyond, this difference is irrelevant since $\exists$SO can express the existence of a total order on 
the structure's universe. A natural characterisation of the complexity class \PTime is known in the form of FO(LFP), first-order
logic with least fixpoints, but this only works when a total order is given in the structure itself 
\cite{DBLP:conf/stoc/Vardi82,DBLP:journals/iandc/Immerman86}. This is known as the Immerman-Vardi Theorem.

Otto studied a fragment of \PTime, namely the set of graph problems that are bisimulation-invariant, denoted \bisim{\PTime},
and showed that this class is captured by the polyadic $\mu$-calculus \polymucalc \cite{Otto99}. This logic extends the 
well-known modal $\mu$-calculus so that its formulas are interpreted in tuples of states rather than single states
\cite{Andersen:PMC:1994,Otto99}. This capturing result is remarkable in that it works for unordered graphs, the reason
being that \polymucalc -- just like $\exists$SO for total orders -- is able to define a partial order on the nodes of
an underlying graph which is actually a total order on the bisimulation quotient. This enables a relatively simple
translation from FO(LFP) into \polymucalc, yielding the lower bound of the capturing result. 
The upper bound follows since fixed formulas of \polymucalc can be model checked in polynomial time 
\cite{Andersen:PMC:1994,Otto99,LL-FICS12}.   

Otto's result raises the question of the possibility to capture the bisimulation-invariant fragments of other complexity
classes, for instance by means of other extensions of the modal $\mu$-calculus. A very powerful logic is \polymucalc's
higher-order extension \PHFL \cite{conf/ifipTCS/LangeL14}. Based on an extension of the Immerman-Vardi Theorem showing
that $k$-\ExpTime for $k >0$ is captured by the order-$(k+1)$ fragment of HO(LFP), higher-order logic with least fixpoints 
\cite{FREIRE201171}, one obtains that \bisim{$k$-\ExpTime} is captured by the order-$k$ fragment of \PHFL 
\cite{DBLP:journals/corr/abs-2209-10311}.\footnote{Note the discrepancy in the traditional ways to assign numerals to orders 
in the two logics considered here: order $1$ in HO(LFP) refers to, like in ``first-order logic'', individual elements and 
order $2$ is for relations like sets thereof. In \PHFL, order $1$ refers to the order of a function, i.e.\ one that takes 
sets of arguments. Hence there is shift by $1$ when going from the predicate logic to the modal logic.} Note that for $k=0$
this is exactly Otto's Theorem with $0$-\ExpTime = \PTime and order-$0$ \PHFL being \polymucalc.

In this paper we aim for a characterisation similar to that of the bisimulation-invariant exponential time-hierarchy, but for the
bisimulation-invariant exponential space hierarchy. Here, the anchor from classical descriptive complexity is not the
Immerman-Vardi Theorem, but the Abiteboul-Vardi-Vianu Theorem stating that the class \PSpace is captured by \FOPFP,
first-order logic with \emph{partial} fixpoints, at least over ordered structures \cite{JACM::AbiteboulVV1997}. Again,
the need for a total order disappears for bisimulation-invariant graph queries (at least for polynomial time and above),
and \FOPFP\ can be translated into $\phfl{1}{}$, the order-$1$ fragment of \PHFL. Note that full $\phfl{1}{}$ is known to 
capture \bisim{\ExpTime}, but the formulas obtained by translating \FOPFP\ happen to fall into a particular fragment,
namely $\trphfl{1}{}$, the \emph{tail-recursive} formulas of this kind. Tail recursion in higher-order modal fixpoint logics
has been studied in the context of space-efficient model checking \cite{conf/rp/BruseLL17}. Here we complete the picture
and extend these known characterisations to capturing results along the entire exponential-space hierarchy, showing that
\bisim{$k$-\ExpSpace} is captured by the tail-recursive fragment of order-$(k+1)$ \PHFL, thus completing the picture in
Fig.~\ref{fig:bisimcomplexity}. 

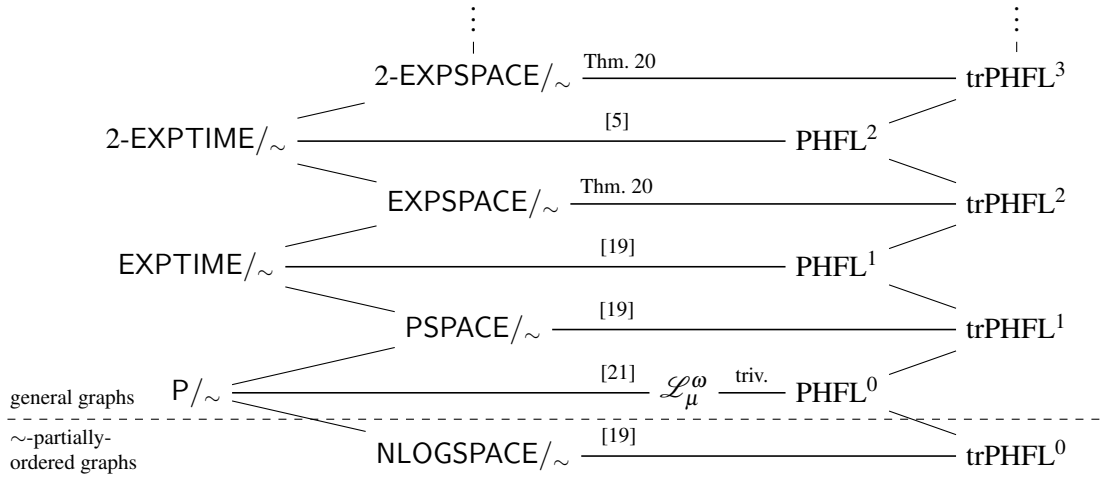
\begin{figure*}[t]
\begin{center}
  \begin{tikzpicture}
     
     \matrix [row sep=1.5mm, column sep=9mm] (cc) {
           & \node (dotsspace) {$\vdots$}; & & & \node (dotstr) {$\vdots$}; \\  
           & \node (bis2expspace) {\bisim{2-\ExpSpace}}; & & & \node (phfl3tr) {$\trphfl{3}{}$}; \\
           \node (bis2exptime) {\bisim{2-\ExpTime}}; & & & \node (phfl2) {$\phfl{2}{}$}; & \\
           & \node (bisexpspace) {\bisim{\ExpSpace}}; & & & \node (phfl2tr) {$\trphfl{2}{}$}; \\
           \node (bisexptime) {\bisim{\ExpTime}}; & & & \node (phfl1) {$\phfl{1}{}$}; & \\
           & \node (bispspace) {\bisim{\PSpace}}; & & & \node (phfl1tr) {$\trphfl{1}{}$}; \\
           \node (bisptime) {\bisim{\PTime}}; & & \node (polymucalc) {\polymucalc}; & \node (phfl0) {$\phfl{0}{}$}; & \\
           & \node (bisnlogspace) {\bisim{\NLogSpace}}; & & & \node (phfl0tr) {$\trphfl{0}{}$}; \\
           };
     
     \tikzset{every node/.style={font=\scriptsize}}
      
     \draw[semithick] (bis2expspace) edge node [above,pos=.1] {Thm.~\ref{thm:capture}} (phfl3tr)
           (bis2exptime) edge node [above,pos=.655] {\cite{DBLP:journals/corr/abs-2511-02596}} (phfl2)
           (bisexpspace) edge node [above,pos=.12] {Thm.~\ref{thm:capture}} (phfl2tr)
           (bisexptime) edge node [above,pos=.665] {\cite{conf/ifipTCS/LangeL14}} (phfl1)
           (bispspace) edge node [above,pos=.16] {\cite{conf/ifipTCS/LangeL14}} (phfl1tr)
           (bisptime) edge node [above,pos=.925] {\cite{Otto99}} (polymucalc)
           (polymucalc) edge node [above] {triv.} (phfl0)
           (bisnlogspace) edge node [above,pos=.1] {\cite{conf/ifipTCS/LangeL14}} (phfl0tr);
           
     \draw[thin] (bisnlogspace) edge (bisptime)
           (bisptime) edge (bispspace)
           (bispspace) edge (bisexptime)
           (bisexptime) edge (bisexpspace)
           (bisexpspace) edge (bis2exptime)
           (bis2exptime) edge (bis2expspace)
           (bis2expspace) edge (dotsspace)
           (phfl0tr) edge (phfl0)
           (phfl0) edge (phfl1tr)
           (phfl1tr) edge (phfl1)
           (phfl1) edge (phfl2tr)
           (phfl2tr) edge (phfl2)
           (phfl2) edge (phfl3tr)
           (phfl3tr) edge (dotstr);
           
     \path (bisptime) -- ++(-2.5,-.35) coordinate (bislinel);
     \draw[dashed,thin] (bislinel) -- ++(14.5,0) coordinate (bisliner);
           
     \path (bislinel) -- ++(-.1,0) node [anchor=south west] {\parbox{2cm}{\scriptsize general graphs}};
     \path (bislinel) -- ++(-.1,0) node [anchor=north west] {\parbox{2cm}{\scriptsize $\sim$-partially- \\ordered graphs}};

  \end{tikzpicture}
\end{center}
\caption{The levels of the exponential time- and space hierarchies are captured by fragments of \PHFL.}
\label{fig:bisimcomplexity}
\end{figure*}

The remainder of this paper is organised as follows. In Sect.~\ref{sec:prel} we recall definitions of \PHFL, 
higher-order logic with fixpoints, and central notions from descriptive complexity. In Sect.~\ref{sec:upper} we
show that queries definable in order-$(k+1)$ tail-recursive \PHFL fall into \bisim{$k$-\ExpSpace}. In Sect.~\ref{sec:lower}
we show that the converse is also true. In Sect.~\ref{sec:concl} we conclude with a few remarks on further work in this
direction. 

\section{Preliminaries}
\label{sec:prel}

Let $\Prop$ and $\Actions$ be finite sets of \emph{propositions} and \emph{actions}. 
A labelled transition system (LTS) is a tuple $\Transsys = (\States, \{\Transition{}{a}{}\}_{a\in\Actions},\Label)$
where $\States$ is a nonempty set of \emph{states},  $\Transition{}{a}{} \subseteq \States  \times \States$ 
is the \emph{transition relation} for each action $a \in \Actions$, and $\Label\colon\States \to 2^{\Prop}$ 
maps a state to the set of propositions that it satisfies.
We write $\Transition{\state}{a}{\statealt}$ instead of $(\state, \statealt) \in \Transition{}{a}{}$.

For $d \ge 1$, a $\dimension$-\emph{pointed LTS} is a  pair $\Transsys, (\state_0,\dotsc,\state_{\dimension-1})$ of an LTS and a $\dimension$-tuple 
of states in it. It is simply called \emph{pointed} when $\dimension$ is clear from the context. If $d=1$ then
this coincides with the standard notion of a pointed LTS. Given some tuple $\overline{\state} = (\state_0,\dotsc,\state_{\dimension-1})$
and $i < \dimension$, we write $\overline{\state}[i]$ to denote $\state_i$, and $\overline{\state}[\statealt/i]$ to denote the tuple 
$(\state_0,\dotsc,\state_{i-1},\statealt,\state_{i+1},\dotsc,\state_{\dimension-1})$.

\subsection{Polyadic Higher-Order Fixpoint Logic}

We assume familiarity with the modal $\mu$-calculus \mucalc \cite{DBLP:conf/focs/Pratt81,Kozen83}. Polyadic Higher-Order Fixpoint Logic \PHFL 
\cite{conf/ifipTCS/LangeL14} extends \mucalc in two ways: (\textsc{i}) \HFL \cite{DBLP:conf/concur/ViswanathanV04} adds 
to it a simply typed $\lambda$-calculus; (\textsc{ii}) polyadic $\mucalc$, denoted by  \polymucalc \cite{Andersen:PMC:1994,Otto99},
is obtained by lifting the interpretation of formulas in states to tuples of states of fixed arity. \PHFL merges both extensions. 
Its formal definition requires a few technicalities.

\paragraph{Types.}

Types govern the syntax of \PHFL, especially those parts that denote functions. Types are derived from the grammar
\[
\tau \Coloneqq \grtype \mid \tau^v \to \tau
\]
where $\grtype$ is a ground type for propositions and $v \in \{\varpos, \varneg, \varboth\}$ is a \emph{variance} denoting 
whether a function is monotonically increasing, monotonically decreasing, or constant in its argument. Variances are 
only needed to check well-typedness; we often do not display them for the sake of readability. 

The \emph{order} $\ORDHFL(\tau)$ of a type $\tau$ is defined via $\ordhfl{\grtype} = 0$ and 
$\ordhfl{\tau_1 \to \tau_2} = \max(\ordhfl{\tau_1}+1, \ordhfl{\tau_2})$.
Types are right-associative, thus every type is of the form $\tau_1 \to \dotsb \to \tau_n \to \grtype$.

Given some LTS $\Transsys = (\States, \{\Transition{}{a}{}\}_{a\in\Actions},\Label)$ and some $0 < \dimension \in \Nat$,
the semantics $\semm{\tau}{\Transsys}$ of a type $\tau$ is a complete lattice, defined inductively as 
\begin{align*}
\semm{\grtype}{\Transsys} &= (2^{\States^{\dimension}}, \subseteq) &
\semm{\tau_2^v \to \tau_1}{\Transsys} &= (\semm{\tau_2}{\Transsys} \to \semm{\tau_1}{\Transsys}, \sqsubseteq_{\tau_2^v \to \tau_1})
\end{align*}
where we tacitly identify a lattice with its domain and where the order $\sqsubseteq_{\tau_2^v \to \tau_1}$ is given by 
\[ f \sqsubseteq_{\tau_2^v \to \tau_1} g \quad \Leftrightarrow \quad \text{for all } x \in \semm{\tau_2}{\Transsys} \text{ we have } 
\left\{ \begin{aligned} 
f(x) \sqsubseteq_{\tau_1} g(x) &, \text{ if } v = \varpos \\
g(x) \sqsubseteq_{\tau_1} f(x) &, \text{ if } v = \varneg \\
f(x) = g(x) &, \text{ if }  v = \varboth.
 \end{aligned}\right.\]
Hence, the semantics of $\grtype$ is the powerset lattice over $\States^{\dimension}$ and the semantics of $\tau_2^v \to \tau_1$
is the lattice of all monotonically increasing, monotonically decreasing or constant functions 
from $\semm{\tau_2}{\Transsys}$ to $\semm{\tau_1}{\Transsys}$,
depending on $v$. Such a set together with the point-wise order above forms a complete lattice if
$\semm{\tau_1}{\Transsys}$ is one. Hence,
monotone functions in such a lattice always have a least and greatest fixpoint due to the 
Knaster-Tarski-Theorem.
We write $\bigsqcup_{\tau_2^v \to \tau_1}$ to denote the join operator in this lattice.

Note that, technically, $\semm{\tau}{\Transsys}$ is also parametrised in $\dimension$. However,
since the dimensionality $\dimension$ is always fixed and usually clear from context, we do not display it to avoid clutter.

\paragraph{Syntax.}

Let $\Prop$ and $\Actions$ be as above and let $\fvars$ and $\lvars$ be finite sets of \emph{fixpoint variables}, resp.\ 
\emph{lambda variables}. We use upper case letters $X, Y, \dots$ for the former and lower case letters $x, y, f, g$ for the latter.

Let $d > 0$. By $[d]$ we denote the set $\{0,\ldots,d-1\}$. The set of -- potentially non-well-formed -- formulas of the $\dimension$-adic fragment of
\PHFL, called $\phfl{}{\dimension}$, is derived via
\[
\varphi \Coloneqq p_i \mid \varphi \vee \varphi \mid \neg \varphi \mid \mudiam{a}{i}\varphi \mid \subst{\sigma}{\varphi} 
\mid \lambda(x \colon \tau).\ \varphi \mid x \mid (\varphi\, \varphi) \mid \mu (X\colon \tau).\, \varphi \mid X
\]
where $x \in \lvars$, $X \in \fvars$, $p \in \Prop$, $a \in \Actions$, $0 \leq i < d$, $\tau$ is a type, and 
$\sigma\colon[d] \to [d]$ is a mapping on so-called \emph{indices}. We assume that standard constants and connectives
such as $\myfalse$, $\wedge, \mubox{a}{i}$ and $\nu (X\colon\tau)$ are available via the obvious dualities. We also
allow ourselves to use more convenient notation for the definition of (more complex) functions and their applications.
For instance, $(\ldots((\varphi\,\psi_0)\,\psi_1)\ldots)\, \psi_{n-1}$ is simply written as 
$\varphi(\psi_0,\ldots,\psi_{n-1})$, and $\lambda (x_0\colon\tau_0).\lambda (x_1\colon\tau_1).\ldots \lambda (x_{n-1}\colon\tau_{n-1}).\psi$
is written as $\lambda (x_0\colon \tau_0,\ldots,x_{n-1}\colon\tau_{n-1}).\psi$, or even $\lambda(x_0,\dotsc,x_{n-1}\colon\tau).\psi$ 
if $\tau = \tau_i$ for all $0 \leq i < n$.

The notions of free and bound variables in a formula are as usual, with $\lambda (x\colon\tau).\varphi$ binding $x$
and with $\mu (X\colon\tau).\ \varphi$ binding $X$. 

Over an LTS $\Transsys$, a $\phfl{}{d}$ formula intuitively defines a set of $d$-tuples in $\Transsys$, or a function that 
transforms such a set into such a set, or into a function etc., depending on the formula's type which we explain shortly. 
Before that we briefly introduce the intuitive meaning of the operators in the syntax. The first five listed in the grammar
above all define a subset of $\States^d$. For example, $p_i$ denotes all tuples such that $p$ holds at their $i$th state.
The operator $\subst{\sigma}{\varphi}$ rearranges the positions in tuples in such a subset. The modality $\mudiam{a}{i}$
expresses that its argument formula holds on a tuple after replacing the $i$th state in it by some $a$-successor. A formula
of the form $\lambda(x\colon\tau^v).\, \varphi$ defines a function that consumes an argument of type $\tau$ and
returns the value of $\varphi$, provided that $x$ is bound to that argument. The typing annotation also claims that this
function is monotonically increasing, monotonically decreasing, or constant in this argument, depending on $v$. A formula 
of the form $(\varphi\, \psi)$ 
denotes the application of the semantics of $\varphi$ to the object defined by $\psi$. Finally, fixpoints
can now also define higher-order functions.

Obviously, not all formulas that can be derived from the above grammar can be given a semantics in a meaningful way;  
consider e.g.\ $(p_1\,p_2)$. Moreover, as is typical in a situation involving least and greatest fixpoints, 
the use of negation has to be restricted, cf.\ the
example $\mu (X\colon\grtype).\ \neg X$. Hence, \PHFL has a type system to filter out formulas that cannot be endowed with
a proper semantics, and to associate a unique type to each subformula. The typing rules are standard. We refer to e.g.\ \cite{DBLP:journals/corr/abs-2209-10311} for the full system.

A formula is said to be of order $k$ if the maximal order of the type of any subformula in $\varphi$
is $k$. By $\phfl{k}{d}$ we denote the set of well-typed formulas in $\phfl{}{d}$ that are of order $k$. 
Note that $\PHFL$ does indeed constitute an extension of other known formalisms, namely
\begin{itemize}
\item $\phfl{k}{1}$ is the same as $\hfl{k}$ for any $k \ge 0$ and, thus $\phfl{}{1} = \HFL$, and in particular
\item $\phfl{0}{1}$ is just the modal $\mu$-calculus \mucalc while
\item $\phfl{0}{d}$ is the $d$-dimensional polyadic $\mu$-calculus \polymucalc.
\end{itemize}

\paragraph{Semantics.}

Let $\Transsys = (\States, \{\Transition{}{a}{}\}_{a\in\Actions},\Label)$ be an LTS and let $\varphi$
be a well-typed formula.
An \emph{environment} is a function $\eta$ that assigns to each fixpoint variable and each lambda variable
of type $\tau$ an element of $\semm{\tau}{\Transsys}$.

Let $d > 0$. The semantics $\sem{\varphi\colon\tau}{\Transsys}{\eta}$ of a $\phfl{}{\dimension}$ formula $\varphi$ of 
type $\tau$, relative to an LTS $\Transsys$ and an environment $\eta$, is an object of $\semm{\tau}{\Transsys}$,
defined recursively as follows.
\begin{align*}
\sem{{p_i \colon \grtype}}{\Transsys}{\eta} & = \{(\state_0,\dotsc,\state_{\dimension-1}) \in \States^\dimension \mid p \in \Label(\state_i)\} \\
\sem{{\varphi_1 \vee \varphi_2 \colon \grtype}}{\Transsys}{\eta} & = \sem{{\varphi_1 \colon \grtype}}{\Transsys}{\eta} \cup
\sem{{\varphi_2 \colon \grtype}}{\Transsys}{\eta} \\
\sem{{\neg \varphi \colon \grtype}}{\Transsys}{\eta} & = \States^\dimension \setminus \sem{{\varphi\colon\grtype}}{\Transsys}{\eta} \\
\begin{split}\sem{{\mudiam{a}{i}\varphi\colon \grtype}}{\Transsys}{\eta} & = \{(\state_0,\dotsc,\state_{\dimension-1}) \in \States^\dimension \mid \text{ ex. } \statealt\text{ s.t. } \\ & \qquad \qquad \qquad (\state_0,\dotsc,\state_{i-1},\statealt,\state_{i+1}, \dotsc, \state_{\dimension-1}) \in \sem{{\varphi \colon \grtype}}{\Transsys}{\eta} \text{ and } \Transition{\state_i}{a}{\statealt}\} \end{split}\\
\sem{{\subst{\sigma}{\varphi}\colon \grtype}}{\Transsys}{\eta} & = \{(\state_0,\dotsc,\state_{\dimension-1}) \in \States^\dimension \mid(\state_{\sigma(0)},\dotsc,\state_{\sigma(\dimension-1)}) \in \sem{{\varphi \colon \grtype}}{\Transsys}{\eta}\} \\
\sem{{\lambda (x^v \colon \tau_2).\ \varphi \colon \tau_2^v \to \tau_1}}{\Transsys}{\eta} &= f \in \semm{\tau_2^v \to \tau_1}{\Transsys} \text{ s.t.~f.a. } y \in \semm{\tau_2}{\Transsys}.\ f(y) 
= \sem{{\varphi \colon\tau_1}}{\Transsys}{\eta[x \mapsto y]} 
\\
\sem{{x \colon \tau}}{\Transsys}{\eta} &= \eta(x) \\
\sem{{(\varphi_1\,\varphi_2) \colon \tau_1}}{\Transsys}{\eta} &= \sem{{\varphi_1 \colon \tau_2^v \to \tau_1}}{\Transsys}{\eta}\, (\sem{{\varphi_2 \colon \tau_2}}{\Transsys}{\eta}) \\
\sem{{\mu (X\colon \tau).\varphi\colon\tau}}{\Transsys}{\eta} &= \bigsqcap_{\tau\to\tau} \{f \in \sem{\tau}{\Transsys}{\eta}\mid \sem{{\varphi\colon \tau}}{\Transsys}{\eta[X\mapsto f]} \sqsubseteq_{\tau} f \} \\
\sem{{X \colon \tau}}{\Transsys}{\eta} &= \eta(X)
\end{align*}
If the type is clear from context, or not important, we simply write $\sem{\varphi}{\Transsys}{\eta}$. We write
$\Transsys, (\state_0,\dotsc,\state_{\dimension-1}) \models_\eta \varphi$ if $\varphi\colon\grtype$ and  
$(\state_0,\dotsc,\state_{\dimension-1}) \in \sem{\varphi}{\Transsys}{\eta}$. We say that two formulas $\varphi$ and
$\psi$ are \emph{equivalent}, written $\varphi \equiv \psi$, if for all $\Transsys$ and all $\eta$ we have 
$\sem{\varphi}{\Transsys}{\eta} = \sem{\psi}{\Transsys}{\eta}$.

It is well-known that the semantics of HFL is invariant under $\beta$-reduction and admits the
fixpoint unfolding principle. Both extend straight-forwardly to \PHFL, i.e.\ we have $(\lambda x.\, \varphi)\; \psi \equiv \varphi[\psi/x]$ and 
$\mu X.\, \varphi \equiv \varphi[\mu X.\, \varphi/X]$ where substitution of free variables 
is defined as usual.

\paragraph*{Tail Recursion.}

Tail recursion in programming refers to a pattern where the recursive call in a function definition happens last,
i.e.\ the result of the function is that of the recursive call. This has two consequences: first, tail-recursive
definitions require less space when executed, since no call stack has to be maintained. Second, tail-recursive
functions are easily converted into while loops. The concept of tail recursion can be brought
to higher-order fixpoint logic, and in fact to general higher-order definitions. It was first defined
in \cite{conf/ifipTCS/LangeL14} and has since been studied intensively for both
modal fixpoint logics \cite{conf/rp/BruseLL17, DBLP:conf/time/BruseLL22} and other other higher-order concepts \cite{DBLP:conf/vmcai/Bruse25}.

Tail recursion for PHFL induces the following two requirements on a formula: (\textsc{i}) all subformulas in
\emph{operand} position, i.e., $\psi_2$ in $\psi_1\ \psi_2$, must be closed w.r.t.\ variables from $\fvars$,
and (\textsc{ii}) formulas with non-monotonic semantics or universal branching, i.e.\ negation or the modal box, must be closed w.r.t. variables from $\fvars$, and as a special case,
one conjunct of a conjunction must also be closed w.r.t.\ these variables. 
We refer to \cite{conf/rp/BruseLL17} for a 
detailed exposition regarding tail-recursive formulas in terms of space-efficient model checking and further details.

\begin{definition}
\label{def:tailrec}
A PHFL formula $\varphi$ is tail recursive if the statement $\fvarsbis \tail \varphi$ can be derived
for some $\fvarsbis \subseteq \fvars$ using the rules in Fig.~\ref{fig:phfl-tail}. 
\end{definition}

\begin{figure}
\begin{mathpar}

\inferrule{ }{\emptyset \tail p_i} 

\inferrule{ }{\emptyset \tail x}

\inferrule{ }{\{X\} \tail X}

\inferrule{\fvarsbis_1 \tail \varphi_1 \\ \fvarsbis_2 \tail \varphi_2}{\fvarsbis_1 \cup \fvarsbis_2 \tail \varphi_1 \vee \varphi_2}

\inferrule{\emptyset \tail \varphi_1 \\ \fvarsbis \tail \varphi_2}{\fvarsbis \tail \varphi_1 \wedge \varphi_2}

\inferrule{\emptyset \tail \psi}{\emptyset \tail \neg \psi}

\inferrule{\fvarsbis \tail \psi}{\fvarsbis \tail \mudiam{a}{i} \psi}

\inferrule{\emptyset \tail \psi}{\emptyset \tail \mubox{a}{i} \psi}

\inferrule{\fvarsbis \tail \psi}{\fvarsbis \tail \{\sigma\} \psi}

\inferrule{\fvarsbis \tail \psi}{\fvarsbis \tail \lambda x.\ \psi}

\inferrule{\fvarsbis \tail \psi_1 \\ \emptyset \tail \psi_2}{\fvarsbis \tail \psi_1\ \psi_2}

\inferrule{\fvarsbis \tail \psi}{\fvarsbis \setminus \{Y\} \tail \mu Y.\ \psi}

\inferrule{\fvarsbis \tail \psi}{\fvarsbis \setminus \{Y\} \tail \nu Y.\ \psi}
\end{mathpar}
\caption{The rules governing tail recursion for PHFL.}
\label{fig:phfl-tail}
\end{figure}

\paragraph{Bisimilarity.} A \emph{bisimulation} $R$ on an LTS 
$\Transsys = (\States, \{\Transition{}{a}{}\}_{a \in \Actions}, \Label)$ is a symmetric relation $R \subseteq \States \times \States$
satisfying the following for all $(s,t) \in R$.
\begin{itemize}
\item $\Label(s) = \Label(t)$,
\item if there is $a \in \Actions$ and $s' \in \States$ s.t.\ $\Transition{s}{a}{s'}$ then there is $t' \in \States$ with
      $\Transition{t}{a}{t'}$ and $(s',t') \in R$,
\item if there is $a \in \Actions$ and $t' \in \States$ s.t.\ $\Transition{t}{a}{t'}$ then there is $s' \in \States$ with
      $\Transition{s}{a}{s'}$ and $(s',t') \in R$.     
\end{itemize} 
Two states $s,t$ are bisimilar, written $s \sim t$, if there is a bisimulation $R$ with $(s,t) \in R$.
Two pointed LTS $\Transsys, (\state_0,\dotsc,\state_{\dimension-1})$ and $\Transsys', (\state'_0,\dotsc,\state'_{\dimension-1})$ are bisimilar if $\state_i \sim \state'_i$ for all $0 \leq i < \dimension$.

Let $d > 0$. A set $T \subseteq \States^d$ is called \emph{bisimulation-invariant} if for all 
$\bar{s} = (s_0,\ldots,s_{\dimension-1}),\bar{t} = (t_0,\ldots,t_{\dimension-1}) \in \States^d$ such that $s_i \sim t_i$ for all $i \in [d]$, we 
have $\bar{s} \in T$ iff $\bar{t} \in T$. The notion of bisimulation-invariance can straight-forwardly be lifted to 
objects of type $\semm{\tau}{\Transsys}$ for types $\tau \ne \grtype$, cf.\ \cite{DBLP:conf/concur/ViswanathanV04}. 

It is well-known that modal logics cannot distinguish bisimilar models, and it is not surprising that the extensions
beyond pure modal logic that are available in \PHFL do not break this property.

\begin{proposition}[\cite{DBLP:conf/concur/ViswanathanV04,conf/ifipTCS/LangeL14}]
\label{prop:bisiminv}
Let $d > 0$, $\Transsys$ be an LTS with state set $\States$ and $\varphi$ be a closed $\phfl{}{d}$ formula of type $\grtype$. Then 
$\semm{\varphi}{\Transsys} \subseteq \States^d$ is bisimulation-invariant.
\end{proposition}

\begin{example}[\cite{Otto99}]
\label{ex:bisim}
A standard example shows that bisimilarity itself is definable in $\phfl{0}{2}$, provided that $\Prop$ and 
$\Actions$ are finite. The $\phfl{0}{2}$ formula
\begin{displaymath}
\varphi_\sim := \nu (X: \grtype). (\bigwedge\limits_{p \in \Prop} p_0 \leftrightarrow p_1) \wedge 
(\bigwedge\limits_{a \in \Actions} \mubox{a}{0}\mudiam{a}{1}X) \wedge \subst{0 \mapsto 1, 1 \mapsto 0}X
\end{displaymath}
is satisfied by a pair $(s,t)$ of some $\Transsys$ iff $s \sim t$. The formula essentially states that $(s,t)$ needs to belong to
the largest set $X$ of states $(s',t')$ that agree on all propositions (first conjunct) and for which $t'$ can match any $a$-transition
out of $s'$, for any $a \in \Actions$, to a pair $(s'',t'') \in X$ (second conjunct). Moreover, $X$ needs to be a symmetric relation (third
conjunct). 
\end{example}

Let $\varphi_\sim^{i,j}$ be the formula that is obtained from $\varphi_\sim$ above by replacing index $0$ with $i$ and index
$1$ with $j$. Note that it can be seen as $d$-ary for any $d$ with $0 \le i,j < d$. When interpreted in a $d$-tuple, it states
that its $i$-th and $j$-th component are bisimilar. This formula will be used several times in the following. 

To exemplify the use of higher-orderness (here: first-order functions) we can use the definability of finite-trace equivalence
in $\phfl{1}{2}$. 
\begin{example}
\label{ex:traceequiv}
Two states $s,t$ are finite-trace equivalent if whenever there is a sequence 
$\Transition{s}{a_0}{s_1}\Transition{}{a_1}{s_2}\Transition{}{a_2}{}$ $\Transition{\ldots}{a_{n-1}}{s_n}$ then there are $t_1,\ldots,t_n$ 
s.t.\ $\Transition{t}{a_0}{t_1}\Transition{}{a_1}{t_2}\Transition{}{a_2}{}\Transition{\ldots}{a_{n-1}}{t_n}$ and vice-versa.
\begin{displaymath} 
\varphi_{\mathsf{fte}} := \big(\nu (F\colon \grtype  \to \grtype ).\ \lambda (x, y \colon \grtype).\ (x \leftrightarrow y) \wedge \bigwedge\limits_{a \in \Actions} F(\mudiam{a}{1}x,\mudiam{a}{2}y))(\mytrue,\mytrue)
\end{displaymath}
is then satisfied by a pair $(s,t)$ iff $s$ and $t$ are finite-trace equivalent in the sense above. The greatest fixpoint in the formula
expresses an infinite conjunction over all finite paths, these can be thought of to be built step-wise via fixpoint unfolding; the $n$th unfolding
expresses that all finite traces of length $n$ are available in the first component of the tuple (and built via $\mudiam{a}{1}$) iff they
are available in the second component (built via $\mudiam{a}{2}$).
For better readability we have omitted the types in the formula. The types of $x,y$ are $\grtype$ and that of $F$ is, consequently,
$\grtype^\varboth \to \grtype^\varboth \to \grtype$. 

We remark that it is easily possible to extend the formula to also check for matching atomic propositions along the traces emerging
from $s$ and $t$, cf.\ \cite{LLVG:TCS:2014} for further examples of how various other process equivalences and preorders can be expressed in
$\phfl{1}{}$.
\end{example}

\subsection{Higher-Order Logic with Partial Fixpoints}

We introduce Higher-Order Logic with Partial Fixpoints (\HOPFP) to make use of the characterisation of $k$-\ExpSpace over the class of ordered structures as the queries definable in order-$(k+1)$ \HOPFP, 
shown in \cite{DBLP:journals/corr/abs-2511-02596} and generalising the characterisation of \PSpace
due to Vardi \cite{DBLP:conf/stoc/Vardi82}.

\paragraph{Types.}
Types for Higher-Order Logic with Partial Fixpoints  are constructed, similar to those for \PHFL, from a single base
type and one constructor: $\tau' \Coloneqq \indtype \mid (\tau',\dotsc,\tau')$. Here, however, $\indtype$ is the type of 
\emph{individuals} (like states rather than sets of states), and the tuple type is used to denote (higher-order) relations. 
In order to distinguish HO types from PHFL types, the former will always be decorated with a prime. We define the order
$\ordho{\tau'_0,\dots,\tau'_{n-1}} = 1 + \max \{ \ordho{\tau'_0},\dotsc,\ordho{\tau'_{n-1}} \}$.\footnote{Recall the discrepancy in the way 
that order is defined in \HOPFP vs.\ \PHFL, leading to an offset by one, as explained in the introduction.} 

Let $\Transsys = (\States, \{\Transition{}{a}{}\}_{a\in\Actions},\Label)$ be an LTS. This induces a set-theoretic 
interpretation of types via $\semm{\indtype}{\Transsys} = \States$ and
$\semm{(\tau'_0,\dots,\tau'_{n-1})}{\Transsys} = 2^{\semm{\tau'_0}{\Transsys} \times \dotsb \times \semm{\tau'_{n-1}}{\Transsys}}$.

\paragraph{Syntax.}
Let $\hovars = \{X,\dotsc\}$ be a countable set of \emph{higher-order variables}, each implicitly equipped with a type. 
Let $\Prop, \Actions$ be sets
of propositions, resp.\ actions. The syntax of \HOPFP formulas is derived from the following grammar:
\[
\varphi \Coloneqq p(X) \mid a(X,Y) \mid X(Y_0,\dotsc,Y_{n-1}) \mid \neg \varphi \mid \varphi \vee \varphi 
\mid \exists X.\ \varphi \mid \big(\pfp(X,Y_0,\dotsc,Y_{n-1}).\ \varphi\big)(Z_0,\dotsc,Z_{n-1})
\]
where $p \in \Prop$, $a \in \Actions$ and $X,Y_0,\dotsc,Y_{n-1},Z_0,\dotsc,Z_{n-1} \in \hovars$. Again, 
operators such as $\wedge$ are available via the obvious dualities. 
Variable $X$ is bound in $\exists X.\ \varphi$, and $X,Y_0,\dotsc,Y_{n-1}$ are bound in 
$\big(\pfp(X,\bar{Y}).\ \varphi\big)(\bar{Z})$ when $\bar{Y} = Y_0,\dotsc,Y_{n-1}$ and $\bar{Z} = Z_0,\dotsc,Z_{n-1}$.
A formula is \emph{well-formed} if each variable is bound at most once, and, moreover, variables occur
only in a way that matches their type. For example, only a variable of type $\indtype$ can occur in a subformula
of the form $p(X)$, and if a subformula of the form $X(Y_1,\dotsc,Y_n)$ or $\big(\pfp(X,Y_0,\dotsc,Y_{n-1}).\ \varphi\big)(Z_0,\dotsc,Z_{n-1})$ occurs, 
then $X$ has type $(\tau'_0,\dotsc,\tau'_{n-1})$
for some $\tau'_0,\dotsc,\tau'_{n-1}$ and each $Y_i$ (and $Z_i$, in the second case) has type $\tau'_i$. Note that there is \emph{no} requirement that, 
in a subformula of the
form $\big(\pfp(X,Y_0,\dotsc,Y_{n-1}).\ \varphi\big)(Z_0,\dotsc,Z_{n-1})$, the variable $X$ must occur only under an even number of
negations in $\varphi$. For a more detailed introduction into Higher-Order Logic including formal ways to define 
well-formedness of formulas we refer to \cite{vanbenthemdoets:1983a}.
An \HOPFP formula $\varphi$ has order $k$ if the order
of the highest type of a variable (note that these are implicitly typed) in $\varphi$ is at most $k$. We write $\hopfp{k}$ for the set of \HOPFP formulas
of order at most $k$.

\paragraph{Semantics.}
Let $\Transsys = (\States, \{\Transition{}{a}{}\}_{a\in\Actions},\Label)$ be an LTS.
A variable assignment $\alpha$ is a function that maps each variable in \hovars\ of type $\tau'$ into
$\semm{\tau'}{\Transsys}$. Its update $\alpha[X \mapsto M]$ is define as usual.

The satisfaction relation $\Transsys, \alpha \models \varphi$ between an LTS $\Transsys$ and a variable assignment $\alpha$ on one side, and an \HOPFP formula $\varphi$ on the other side
is defined inductively: 
\begin{align*}
\Transsys, \alpha \models p(X) &\text{ iff } p \in \Label(\alpha(X)) \\
\Transsys, \alpha \models a(X,Y) &\text{ iff } \Transition{\alpha(X)}{a}{\alpha(Y)} \\
\Transsys, \alpha \models X(Y_0,\dotsc,Y_{n-1}) &\text{ iff } (\alpha(Y_0),\dotsc,\alpha(Y_{n-1})) \in \alpha(X) \\
\Transsys, \alpha \models \neg \varphi &\text{ iff } \Transsys, \alpha \not \models \varphi \\
\Transsys, \alpha \models \varphi_1 \vee \varphi_2 &\text{ iff } \Transsys, \alpha \models \varphi_1 \text{ or } \Transsys, \alpha \models \varphi_2 \\
\Transsys, \alpha \models \exists X.\ \varphi &\text{ iff ex.\ } d \in \semm{\tau'}{\Transsys} \text{ s.t.\ } \Transsys, \alpha[X \mapsto d] \models \varphi \\
\Transsys, \alpha \models \big(\pfp(X,Y_0,\dotsc,Y_{n-1}).\ \varphi\big)(Z_0,\dotsc,Z_{n-1}) &\text{ iff } (\alpha(Z_0),\dotsc,\alpha(Z_{n-1})) \in \PFP(X,Y_0,\dotsc,Y_{n-1}, \varphi)
\end{align*}
The term $\PFP(X,Y_0,\dotsc,Y_{n-1}, \varphi)$ in the last case is defined via mutual recursion with the satisfaction
relation as follows.
Let $\varphi$ be an \HOPFP formula with free variables in $X, Y_0,\dotsc,Y_{n-1}$ and let $\tau' = (\tau'_0,\dotsc,\tau'_{n-1})$ be the type of $X$ in $\varphi$
while $\tau'_i$ is the type of $Y_i$ in $\varphi$ for $0 \leq i < n$. Note that $\tau'$ is a set type here. 
Then, given some variable assignment $\alpha$, the formula $\varphi$ defines a
function $f \colon \semm{\tau'}{\Transsys} \to \semm{\tau'}{\Transsys}$ via
\[f \colon M \mapsto \{(m_0,\dotsc,m_{n-1}) \mid \Transsys, \alpha[X \mapsto M, Y_0 \mapsto m_0, \dotsc, Y_{n-1} \mapsto m_{n-1}] \models \varphi\}.\]
When constructing the sequence $\emptyset, f(\emptyset), f^2(\emptyset),\dotsc$,
this sequence can converge, i.e.\ $f^{k+1}(\emptyset) = f^k(\emptyset)$ for some $k \ge 0$. Otherwise it diverges, i.e.\ each element in this sequence is different
to the previous one. On finite structures this amounts to entering a cycle. The term $\PFP(X,Y_0,\dotsc,Y_{n-1}, \varphi)$ denotes the point of convergence in the first 
case, also known as a fixpoint, and $\emptyset$ in the divergence case. Hence, the iteration can reach a fixpoint naturally, or it can artificially be given a default
value. This is why $\PFP$ is called a \emph{partial fixpoint} operator. Note that we suppress $\Transsys$ and $\alpha$ here since they will always be clear from context.
The mutual recursion between $\models$ and $\PFP(\dots)$ is well-defined since it is always strictly descending
in the formula in question.

\begin{example}
\label{ex:pfp}
Let $\M{A} = (Q, \delta, I, F)$ be a one-letter NFA that is total, i.e., for all $q \in Q$, there is $q'$
such that $(q,a,q') \in \delta$. Then $\M{A}$ can be seen as an LTS with one action $a$ and two propositions $i, f$, induced by $I$ and $F$. 
Consider the \FOPFP\ formula $\mathsf{step}$ with free variables $X$ of type $(\indtype)$ and $Y$ of type $\indtype$, given by
\begin{align*}
\mathsf{step}(X, Y) =\ & \quad \big(i(Y) \wedge \forall Y'.\ \neg X(Y')\big) \\
                      & \vee \big(  X(Y) \wedge \exists Y'_1.\ X(Y'_1) \wedge \forall Y'_2. (X(Y'_2) \rightarrow \neg f(Y'_2))\big)
											 \\
											& \vee \big(\exists Y'_1.\ a(Y'_1, Y) \wedge X(Y'_1) \wedge \exists Y'_2.\ X(Y'_2) \wedge f(Y'_2)\big).
\end{align*}
On $\M{A}$, given some set $M$ it returns the set $M' = \{\state \mid \M{A}, [X\mapsto M, Y \mapsto \state \models \mathsf{step}\}$. 
This set can be thought of as being defined via a case distinction:
\begin{itemize}
\item If $M$ is empty, then $M'$ is the set if states that satisfy $i$.
\item If $M$ is nonempty and contains no state that does satisfy $f$, then $M' = M$.
\item If $M$ is nonempty and some state in $M$ satisfies $f$, then $M'$ is the set of $a$-successors of the
states in $M$.
\end{itemize}
Clearly, $\mathsf{step}$ is not monotonic in general.
Now consider the formula
\[
\varphi_{\mathsf{nonuniv}} \coloneqq \exists Z.\ (\pfp(X,Y).\  \mathsf{step})(Z) \wedge \forall Z.\ (\pfp(X,Y)\ \mathsf{step})(Z) \rightarrow \neg f(Z).
\]
It is true iff $\M{A}$ is non-universal, i.e., if there is $k$ such that $a^k$ is not accepted by $\M{A}$.
This can be seen as follows: the set $X$ of the partial fixpoint definition contains the states reachable by
$a^0, a^1, \dots$ in sequence, as long as at least one of them is accepting. 
The first case above initializes $X$ with the set of initial states after 
the first iteration. Since $\M{A}$ is total, the set never runs empty afterwards. 
Then the third clause keeps passing to successors as long as the previous set contains
at least one final state. Should the set in $X$ contain only non-final states, the second clause
fixes this set forever, reaching a partial fixpoint. If this never happens, there are two cases:
either the sequence also reaches a partial fixpoint at some point, or it never stabilizes. In the latter
case, the partial fixpoint is the empty set. Since $\varphi_{\mathsf{nonuniv}}$  checks for a nonempty partial
fixpoint that contains only non-final states, it is satisfied iff $\M{A}$ is non-universal\footnote{
We should note that the problem of one-letter NFA non-universality is known to be NP-complete, while
\FOPFP\ captures PSPACE over ordered structures. Hence, this example
illustrates typical patterns in conjunction with partial fixpoints, but does not fully exhaust the 
expressive power of \FOPFP.}.
\end{example}

\paragraph*{Least Fixpoints.}

If, in a formula of the form $ \big(\pfp(X,Y_0,\dotsc,Y_{n-1}).\ \varphi\big)(Z_0,\dotsc,Z_{n-1})$,
the fixpoint variable $X$ occurs only under even numbers of negations, and $\varphi$ contains no further
fixpoint definitions, then the function
\[f \colon M \mapsto \{(m_0,\dotsc,m_{n-1}) \mid \Transsys, \alpha[X \mapsto M, Y_0 \mapsto m_0, \dotsc, Y_{n-1} \mapsto m_{m-1}] \models \varphi\}\]
defined by $\varphi$ is monotone, and, hence, always reaches a fixpoint by the Knaster-Tarski and Kleene 
fixpoint theorems. Moreover, this is also the case if $\varphi$ contains further fixpoint definitions, but
these are monotonic due to restricted occurrences of \emph{their} fixpoint variable. Formulas of
this kind form \emph{Higher-Order Logic with Least Fixpoints} (\HOLLFP), which is also
stratified w.r.t.\ type-theoretic order, just as \HOPFP.

\subsection{Descriptive Complexity}

We recall some standard definitions from descriptive complexity theory. A \emph{query} (of dimension $d$) is 
a set $Q$ of pairs $(\Transsys,(s_0,\ldots,s_{\dimension-1}))$ s.t.\ each $\Transsys$ is finite. It is 
expressed by the \HOPFP formula 
$\varphi$ with free variables $X_0,\dotsc,X_{\dimension-1}$, if $Q = \{(\Transsys, (\state_0,\dotsc,\state_{\dimension-1})) \mid \Transsys, [X_0 \mapsto \state_0,\dotsc, X_{\dimension-1}\mapsto \state_{\dimension-1}] \models \varphi \}$.

Let $k \ge 0$. The complexity class $k$-\ExpSpace is defined as $\mathit{DSPACE}(2^{n^{\mathcal{O}(1)}}_k)$ where 
$2^m_0 := m$ and $2^m_{k+1} := 2^{2^{m}_k}$. Note that $0$-\ExpSpace equals $\PSpace$.

For a complexity class $\mathcal{C}$, a $\mathcal{C}$-query is one that can be decided within the resource bounds given by
$\mathcal{C}$. We write \bisim{$\mathcal{C}$} for the complexity class of $\mathcal{C}$-queries that are bisimulation-invariant. 

We say that a logic $\mathcal{L}$ \emph{captures} a complexity class $\mathcal{C}$ if the model checking problem for $\mathcal{L}$
is in $\mathcal{C}$ and each $\mathcal{C}$-query can be expressed in $\mathcal{L}$. 
 
The \PSpace-queries 
\emph{over ordered structures} are those expressible in $\hopfp{1}$, resp.\ first-order logic with partial fixpoints \cite{DBLP:conf/stoc/Vardi82}. This can be generalised as follows.

\begin{proposition}[\cite{DBLP:journals/corr/abs-2511-02596}]
\label{prop:IV-gen}
For each $k \ge 0$, $\hopfp{k+1}$ captures $k$-\ExpSpace over ordered structures.
\end{proposition}

The ordering is only important for the case of $k=0$ as such an ordering can be defined in second-order logic, i.e.\ as soon as
$k \ge 1$.
The first capturing result for a bisimulation-invariant class is given by Otto's Theorem.

\begin{proposition}[\cite{Otto99}]
\label{prop:otto}
The polyadic modal $\mu$-calculus \polymucalc, or, equivalently $\phfl{0}{}$, captures \bisim{\PTime},
or equivalently, \bisim{$0$-\ExpTime}.
\end{proposition}

We give a very brief sketch of the proof, since we re-use the general pattern. The upper bound, i.e.\
that model checking for \polymucalc is in $\PTime$, is by verification. Bisimulation-invariance
is by Prop.~\ref{prop:bisiminv}.

The lower bound is obtained by three observations: (\textsc{i}) the formula $\varphi_\sim$ that 
expresses bisimilarity of two states can be modified to turn \emph{non-bisimilarity} into a total
order of the bisimulation classes, whence w.l.o.g.\ the structures over which $\phfl{}{}$, and, 
hence, \polymucalc, is defined, are always totally ordered, and this order is \polymucalc-definable. 
Moreover, (\textsc{ii}) on totally
ordered structures, the Immerman-Vardi Theorem \cite{DBLP:conf/stoc/Vardi82,DBLP:journals/iandc/Immerman86} posits 
that the polynomial-time queries are exactly those expressed by first-order logic with least fixpoints. 
Finally, (\textsc{iii}) \emph{bisimulation-invariant} first-order logic with least fixpoints can
be translated into \polymucalc. The main difficulty here is to deal with first-order quantification,
which has no natural counterpart in the modal world. This can be replaced by reachability in the
bisimulation-invariant setting. It follows that every bisimulation-invariant $\PTime$ query
can be expressed in $\polymucalc$.

Using reasoning along similar lines, Otto's Theorem has also been generalised to the full exponential-time hierarchy.

\begin{proposition}[\cite{Otto99,conf/ifipTCS/LangeL14,DBLP:journals/corr/abs-2209-10311}]
\label{prop:capturebisimexptime}
For each $k \ge 0$, $\phfl{k}{}$ captures \bisim{k-\ExpTime}.
\end{proposition}

\section{Upper Bounds}
\label{sec:upper}
 
The upper bound of the capturing of levels of the bisimulation-invariant exponential-space
hierarchy is established by showing
that all queries definable in tail-recursive $\phfl{k+1}{d}$ can be evaluated in $k$-\ExpSpace for any $d > 0$. 
We do so by extending the reduction of the model-checking problem for a 
polyadic logic to that of its monadic fragment, cf.\ \cite{LL-FICS12, DBLP:journals/corr/abs-2209-10311}. 
Note that tail-recursive $\phfl{k}{1}$ equals tail-recursive $\hfl{k}$ i.e.\ non-polyadic Higher-Order Fixpoint 
Logic of formulas of type order at most $k$. The complexity of model checking these fragments is known. 

\begin{proposition}[\cite{conf/rp/BruseLL17}]
\label{prop:hflmodelcheck}
Let $k \ge 0$. Model checking for tail-recursive $\hfl{k+1}$ is $k$-\ExpSpace-complete.
\end{proposition} 
The corresponding result for tail-recursive $\phfl{k}{d}$ follows
via a polynomial reduction.

\begin{theorem}
\label{thm:upper}
Let $k \geq 0,d \ge 1$. The model checking problem for $\phfl{k+1}{d}$ is in $k$-\ExpSpace.
\end{theorem}

\begin{proof}
Let $k \geq 0,d \ge 1$ and an LTS 
$\Transsys = (\States, \{\Transition{}{a}{}\}_{a\in\Actions},\Label)$ over $\Actions$ and $\Prop$. Its $d$-product $\Transsys^d$ 
(cf.\ \cite{LL-FICS12,DBLP:journals/corr/abs-2209-10311}) over the action set 
$\Actions^d := \{a_i \mid a \in \Actions, i \in [d] \} \cup \{ \sigma \mid \sigma: [d] \to [d] \}$ and atomic propositions
$\Prop^d := \{ q_i \mid q \in \Prop, i \in [d] \}$ is $(\States^d, \{\Transition{}{x}{}\}_{x \in \Actions'_d}, \Label')$ 
where, for all $q \in \Prop$, $i \in [d]$, $a \in \Actions$, $s_0,\ldots,s_{d-1},t_0,\ldots,t_{d-1} \in \States$ we have
\begin{itemize}
\item $q_i \in \Label'(s_0,\ldots,s_{d-1})$ iff $q \in \Label(s_i)$,
\item $\Transition{(s_0,\ldots,s_{d-1})}{a_i}{(t_0,\ldots,t_{d-1})}$ iff $\Transition{s_i}{a}{t_i}$ and $t_j=s_j$ for all $j \ne i$,
\item $\Transition{(s_0,\ldots,s_{d-1})}{\sigma}{(t_0,\ldots,t_{d-1})}$ iff $t_j = \sigma(s_j)$ for all $j \in [d]$.
\end{itemize}

We now translate a tail-recursive $\phfl{k+1}{d}$ formula $\varphi$ inductively into a tail-recursive $\phfl{k+1}{1}$ formula $\widehat{\varphi}$ as follows.
The operation $\widehat{\cdot}$\ acts homomorphically on all operators apart from the following three cases.
\begin{displaymath}
\widehat{\subst{\sigma}\psi} := \mudiam{\sigma}{}\widehat{\psi} \enspace, \quad 
\widehat{\mudiam{a}{i}\psi} := \mudiam{a_i}{}\widehat{\psi} \enspace, \quad 
\widehat{p_i} := p_i \ . 
\end{displaymath}
Clearly this maps tail-recursive polyadic formulas to tail-recursive monadic formulas of the 
same order.

The above equality in syntax of $p_i$ is intended because of the
following connection: the $d$-tuple $(s_0,\ldots,s_{d-1})$ of states in $\Transsys$ satisfies the $\phfl{}{d}$ formula 
$p_i$ (formed as proposition $p$, indexed by $i$) iff the state $(s_0,\ldots,s_{d-1})$ of $\Transsys^d$ satisfies the 
$\phfl{}{1}$ proposition $p_i$.
 
From \cite{LL-FICS12,DBLP:journals/corr/abs-2209-10311} we have the following (via induction
on the syntax of tail-recursive \PHFL): for all $\phfl{}{d}$ formulas $\varphi$, all 
$s_0,\ldots,s_{d-1} \in \States$ and all environments $\eta$ we have: $(s_0,\ldots,s_{d-1}) \in \sem{\varphi}{\Transsys}{\eta}$ 
iff $(s_0,\ldots,s_{d-1}) \in \sem{\widehat{\varphi}}{\Transsys^d}{\eta}$. Hence,
the reduction is correct, and it is not hard to see that it also works in polynomial time. Thus, by Prop.~\ref{prop:hflmodelcheck}, model checking $\phfl{k+1}{d}$ is also in $k$-\ExpSpace.
\end{proof}

\section{Lower Bounds}
\label{sec:lower}

We now provide the counterpart of Thm.~\ref{thm:upper} and show that problems on level $k$ of the exponential-space hierarchy that are also bisimulation-invariant are definable in tail-recursive order-$(k+1)$ \PHFL.

\subsection{The General Idea}

Since, by Prop.~\ref{prop:IV-gen} (\cite{DBLP:conf/stoc/Vardi82,DBLP:journals/corr/abs-2511-02596}), for each $k \geq 1$ we have that
 $\hopfp{k+1}$ captures $k$-\ExpSpace, it suffices to show that each bisimulation-invariant $\hopfp{k+1}$ query can
be translated into an equivalent tail-recursive $\phfl{k+1}{}$ query in order to obtain the missing lower bound. 
This follows the pattern of Otto's Theorem in which first-order logic with least fixpoints is translated into \polymucalc \cite{Otto99}.
There, first-order variables are replaced by polyadicity, least second-order fixpoints become least fixpoints of \polymucalc, and first-order quantification
is replaced by reachability.

An extension towards higher-order logic with least fixpoints was given in \cite{conf/ifipTCS/LangeL14,DBLP:journals/corr/abs-2209-10311}.
The main difficulty is dealing with higher-order objects, which have no obvious counterpart, both for quantification and also
for general use. This requires defining a translation from types of higher-order logic to suitable counterparts for PHFL. In order to 
de-clutter notation, we make the following simplification.
Let $w\geq 2$ be fixed but arbitrary. 
Define $\tau'_{w,k}$ as $\tau'_{w,1} = \indtype$, $\tau'_{w,i+1} = (\tau'_{w,i},\dotsc,\tau'_{w,i})$ with
$w$ many repetitions of $\tau'_{w,i}$. 
\begin{lemma}
\label{lem:homo}
If $\varphi \in \hopfp{k}$ defines a query $Q$, then there is $\varphi' \in \hopfp{k}$ that defines 
the same query, but the only types used in $\varphi'$ are $\tau'_{w,0},\dotsc,\tau'_{w,k}$ for some $w$.
\end{lemma}
This is proved in the same way as Lemma 13 in \cite{DBLP:journals/corr/abs-2209-10311} where an analogous result is stated for HO$^k$(LFP).
Hence, from now on, all $\hopfp{}$ queries use only the types introduced above.

\subsection{Type Correspondence}

Let $\Transsys, (\state_0,\dotsc,\state_{\dimension-1})$ be a pointed LTS. Clearly, it is bisimilar
to any version of $\Transsys$ where all states not reachable from one of the $\state_i$ (with $0 \leq i < \dimension$) are removed. 
Hence, in the bisimulation-invariant setting, we can always assume all states
of a given LTS to be reachable from one of the $\state_i$, if $(\state_0,\dotsc,\state_{\dimension-1})$ is
the designated input tuple.

In the following, we develop $\phfl{}{D}$ formulas for some dimension $D = 2w + d + 2$ where
$d$ is the width of the input tuple in the query in question and $w \geq d$ depends on the types used in the query to be translated (to be made precise
below). We only work with tuples of the form $(\statealt_0,\dotsc,\statealt_{D-d},\state_{0},\dotsc,\state_{\dimension-1})$, i.e., 
with tuples where the last $d$ elements contain a copy of the input tuple $(\state_0,\dotsc,\state_{\dimension-1})$. 

\begin{definition}
\label{def:good}
Let $\Transsys$ be an LTS and let $w \geq d$ and let $D = 2w+d+2$ be fixed, and let $(\state_0,\dotsc,\state_{\dimension-1})$ be
a dedicated tuple of states such that all states of $\Transsys$ are reachable from one of
the $\state_i$. We call the following \emph{good} w.r.t. $(\state_0,\dotsc,\state_{\dimension-1})$. For brevity we drop the
part that states which tuple something is good for. It will also be clear from context.
\begin{itemize}
\item A tuple $\overline{\state}$ is good if it is of the form $\overline{\state} = (\statealt_0,\dotsc,\statealt_{D-d-1},\state_{0},\dotsc,\state_{\dimension-1})$.
\item A set is good if it is of the form $M \times 2^{\States^{D-w-d}} \times \{\state_0\} \times \dotsb \times \{\state_{d-1}\}$ where $M \subseteq 2^{\States^w}$.
\item A function $f$ of type $\tau_1 \to \tau_2$ is good if, on good inputs, it returns either $\States$ or $\emptyset$.
\item An environment $\eta$ is good if it maps all variables to good sets or functions, and 
in particular if it maps the ground-type variable $\mathsf{e}$ to $\States^{D-d} \times \{\state_0\} \times \dotsb \times \{\state_{\dimension-1}\}$.
\end{itemize}
\end{definition}
The variable $\mathsf{e}$ is used in the following to restrict to good sets by intersecting with $\eta(\mathsf{e})$. 

The intuition
here is that good sets are those that are nontrivial only in the first $w$ elements, full in the next $D-w-d$ elements, and have their last $d$
components fixed to the input tuple. Finally, good functions can be seen as characteristic functions
of a set, i.e., they only return ``true'' (indicated by $\States$) or ``false'' (indicated by the empty set).

Let $\tau_{w,k}$ be inductively defined via $\tau_{w,0} = \grtype$ and 
$\tau_{w,i+1} = \tau_{w,i} \to \dotsb \to \tau_{w,i} \to \grtype$
where $\tau_{w,i}$ is repeated $w$ many times, as defined above. 
Given $\Transsys$, let $\semm{\tau^\circ_{w,0}}{\Transsys}$ be the restriction
of $\semm{\tau^\circ_{w,i}}{\Transsys}$ to good sets, resp.\ functions.

We now want to emulate the type $\tau'_{w,k}$ with $k>1$ by $\tau^\circ_{w,k-2}$. The type $\tau'_{w,1}$ will
be handled by polyadicity as in \cite{Otto99}. 
Let $\Transsys$ be an LTS with state set $\States$, let $(\state_0,\dotsc,\state_{d-1})$ be an input tuple
such that all states are reachable from one of the $\state_i$, and let $D = 2w + d + 2$ for some $w \geq 1$. For each $k \geq 2$, we define a translation
$\SEMTRANS_k^{\Transsys}\colon\semm{\tau'_{w,k}}{\Transsys} \to \semm{\tau^\circ_{w,k-2}}{\Transsys}$ via 
$\semtrans{M}{\Transsys}{2} = M \times \States^{D-w-d} \times \{\state_0\} \times  \dotsb \times \{\state_{d-1}\}$ and, for $i \geq 2$, 
$\semtrans{M}{\Transsys}{i+1} = f \in \semm{\tau_{w,i-2}}{\Transsys}$ s.t.\ 
\begin{align*}
f(f_1,\dotsc,f_n) = \left\{
\begin{aligned}
\States, &\text{ if } f_j = \semtrans{M_i}{\Transsys}{i} \text{ f.a. } 0 \leq j \leq w \text{ and } (M_0,\dotsc,M_{w-1}) \in M \\
\emptyset, &\text{ otherwise.}
\end{aligned}\right.
\end{align*}
Hence, a variable assignment $\alpha$ such that all variables are of types $\tau'_{w,2},\dotsc,\tau'_{w,k}$, induces an
environment $\eta_\alpha$ with definitions for all those variables of order $\geq 2$ via $\eta_\alpha(x) = \semtrans{X}{\Transsys}{i}$
where $i$ is the order of $X$. Note that we use lower-case variables on the PHFL side as these are all lambda variables. 
We refrain from adding an explicit translation for variables and simply use the lower-/upper-case correspondence.

\subsection{The Translation}

We now adapt the translation given in \cite{Otto99} for polynomial-time queries, and extended in \cite{DBLP:journals/corr/abs-2209-10311} to
$k$-fold exponential-time queries, to the exponential-space hierarchy.
The centrepiece is a  syntactical translation $\SYNTRANS$ from $\hopfp{k}$ with $k \geq 1$ into tail-recursive $\phfl{k+1}{D}$ for suitable $D = 2w + d + 2$, where w.l.o.g.\ $w \geq d$.
It is designed to satisfy the following invariant:
\begin{equation}\label{eq:invariant}\tag{$\star$}
\text{$\Transsys, \alpha \models \varphi$ iff $(\alpha(X_0),\dotsc,\alpha(X_{w-1}),\statealt_{w},\dotsc,\statealt_{D-d-1},\state_0,\dotsc,\state_{d-1}) \in \sem{\syntrans{\varphi}}{\Transsys}{\eta_\alpha}$ f.a.\ $\statealt_{w},\dotsc,\statealt_{D-d-1}$.}
\end{equation}
The assumption here is that the free first-order variables of $\varphi$ are exactly $X_0,\dotsc,X_{w-1}$.

\paragraph*{First-Order Quantification.}

Before we give the basic parts of the translation, recall that \cite{Otto99} defines a formula $\varphi_<$ which compares two elements
w.r.t.\ the \polymucalc-definable order $<$. We assume that this compares the elements at positions
$D-d-2$ and $D-d-1$, i.e., $\overline{\state} \in \sem{\varphi_<}{\Transsys}{\eta}$ iff
$\overline{\state}[D-d-2] < \overline{\state}[D-d-1]$. Let $\sigma^{<_i}$ be defined via $\sigma^{<_i}(D-d-2) = i$, $\sigma^{<_i}(D-d-1) = i+w$ and $\sigma^{<_i}(j) = j$
otherwise.  Then $\overline{\state} \in \sem{\subst{\sigma_i}{\varphi_<}}{\Transsys}{\eta}$ iff
$\overline{\state}[i] < \overline{\state}[i+w]$.

Define the substitution operator $\sigma^{i \leftarrow j}$ via $\sigma^{i \leftarrow j}(i) = j$ and $\sigma^{i \leftarrow j}(i') = i'$ if $i' \not = i$. 
Now let $\Actions$ be a set of actions, let $D > d \geq 2$ be given and let $i < D - \dimension$. 
Consider the formulas
\begin{align*}
\exists_i\ \varphi & \coloneqq \bigvee_{j = D - \dimension}^{D-1} \subst{\sigma^{i\leftarrow j}}{\big(\mu (X\colon\grtype).\ \varphi \vee \bigvee_{a \in \Actions} \mudiam{a_i}{}X\big)} \hspace*{1cm}
\forall_i\ \varphi & \coloneqq \bigwedge_{j = D - \dimension}^{D-1} \subst{\sigma^{i\leftarrow j}}{\big(\nu (X\colon\grtype).\ \varphi \wedge \bigwedge_{a \in \Actions} \mubox{a_i}{}X\big)}
\end{align*}
for some $\varphi \in \phfl{}{\dimension}$. 
\begin{lemma}{(\cite{Otto99})}
\label{lem:ind-quant}
Let $\Transsys$ be an LTS over $\Actions$, and
let $\eta$ and $\overline{\state}$ be good.
Then $\overline{\state} \in \sem{\exists_i\varphi}{\Transsys}{\eta}$ iff there is  a state $\statealt$ such that $\overline{\state}[\statealt/i] \in \sem{\varphi}{\Transsys}{\eta}$, and $\overline{\state} \in \sem{\forall_i\varphi}{\Transsys}{\eta}$ iff  $\overline{\state}[\statealt/i] \in \sem{\varphi}{\Transsys}{\eta}$
for all $\statealt$.
\end{lemma}
Hence, $\exists_i\varphi$ mimics existential first-order quantification, and the tuple
obtained by quantification is also good. Moreover, $\exists_i\varphi$ and $\forall_i\varphi$ are tail-recursive and
closed w.r.t.\ fixpoint variables if $\varphi$ is so.

\paragraph*{The Basic Translation.}

We now give the basic parts of $\SYNTRANS$ which are directly borrowed from \cite{Otto99}:
\begin{align*}
\syntrans{p(X_i)} &\coloneqq p_i &
\syntrans{\varphi_1 \vee \varphi_2} &\coloneqq \syntrans{\varphi_1} \vee \syntrans{\varphi_2} \\
\syntrans{a(X_i,X_j)} &\coloneqq \mudiam{a}{i}{\varphi^{i,j}_\sim} &
\syntrans{\varphi_1 \wedge \varphi_2} &\coloneqq \syntrans{\varphi_1} \wedge \syntrans{\varphi_2} \\
\syntrans{\neg \varphi} &\coloneqq \neg \syntrans{\varphi} &
\syntrans{\exists (X_i\colon\indtype).\ \varphi} &\coloneqq \exists_i.\ \syntrans{\varphi} \\
\syntrans{X(X_{i_0},\dotsc,X_{i_{w-1}})} &\coloneqq \subst{\sigma}x \text { if } (X\colon\tau'_{w,1}) &
\syntrans{\forall (X_i\colon\indtype).\ \varphi} &\coloneqq \forall_i.\ \syntrans{\varphi}
\end{align*}
$\varphi^{i,j}_\sim$ is the formula from Ex.~\ref{ex:bisim}, resp.\ \cite{Otto99} that holds iff the states in positions
$i,j$ are bisimilar, and $\sigma$ is defined via $\sigma(j) = i_j$ for $0 \leq j \leq w$ and
$\sigma(k) = k$ for $k \geq w$. The substitions $\sigma_{i,j}$, resp.\ $\sigma$ serve to swap the elements $i$ and $j$ to positions
$\dimension-1$ and $\dimension$, resp.\ to reorder the elements according to the variable order used on the left
side of the translation.

\begin{lemma}
\label{lem:lower-basic}
Assume that $\syntrans{\varphi}, \syntrans{\varphi_1}$ and $\syntrans{\varphi_2}$ are tail-recursive, have no free fixpoint variables,
 and satisfy the invariant \eqref{eq:invariant}.
Then the translations of $p(X_i),  a(X_i, X_j),  \varphi_1 \vee \varphi_2, \varphi_1 \wedge \varphi_2, \neg \varphi$
as well as $\exists (X_i\colon\indtype).\ \varphi, \forall (X_i\colon\indtype).\ \varphi$ and $X(X_{i_0},\dotsc,X_{i_{w-1}})$
for $X\colon\tau'_{w,1}$
are tail recursive, have no free fixpoint variables, and satisfy the invariant \eqref{eq:invariant}, too.
\end{lemma}
See \cite{Otto99} for the correctness part. The claim on tail-recursiveness is by verification. Note that, without asking for negation normal form, it is crucial that all 
formulas be fixpoint closed, for otherwise occurrences of negation could break tail-recursiveness.

\paragraph*{Generalised Quantification.} We now sketch how higher-order quantification is handled in $\SYNTRANS$. This is a condensed
version of the presentation in \cite{DBLP:journals/corr/abs-2209-10311}.

\begin{restatable}{lemma}{quantifiersets}
\label{lem:quantifier-sets}
Let $\varphi$ be tail recursive and fixpoint closed, let $\Transsys$ be an LTS and let $\eta$ be good. Then there are functions
$\exists^{(w)} x.\ \varphi$ and $\forall^{(w)} x.\ \varphi$ such that
$\overline{\state} \in \sem{\exists^{(w)} x.\ \varphi}{\Transsys}{\eta}$ iff  there is a good $M \subseteq \States^D$ such that $\overline{\state} \in \sem{\varphi}{\Transsys}{\eta[x \mapsto M]}$, and
$\overline{\state} \in \sem{\forall^{(w)} x.\ \varphi}{\Transsys}{\eta}$ iff $\overline{\state} \in \sem{\varphi}{\Transsys}{\eta[x \mapsto M]}$ for all good $M \subseteq \States^D$.
 \end{restatable}
\begin{proof}
We sketch the construction.
Let $\sigma_{\rightarrow w}$ be defined by 
$\sigma_{\rightarrow w}(i) = w+i$ for $i < w$ and $\sigma_{\rightarrow w}(j) = j$ for $j \geq w$. 
Then $\overline{\state} \in \sem{\subst{\sigma_{\rightarrow w}}{\varphi}}{\Transsys}{\eta}$ iff
$\overline{\state}[\overline{\state}[0]/w,\dotsc,\overline{\state}[w-1]/2w-1] \in \sem{\varphi}{\Transsys}{\eta}$. 
The intended use here is to shift the first $w$ elements of a formula to the right to make room for another tuple at the first $w$
positions such that the two tuples can be compared lexicographically.

  Now consider the following formulas.
\begin{align*}
\varphi_<^w          &\coloneqq \bigvee_{i=0}^{w-1} \subst{\sigma_i}{\varphi_<} \wedge \bigwedge_{j=0}^{i-1} \subst{\sigma_j}{\neg \varphi_<} \\
\varphi_<^{(w)}(x,y) &\coloneqq \exists_{0}\dotsc\exists_{w-1}.\ y \wedge \neg x
                        \wedge \subst{\sigma_{\rightarrow w}}{\big(\forall_{0}\dotsc\forall_{w-1}. \varphi_<^w \rightarrow x \rightarrow y\big)}. \\
\nxt^{(w)}(x)        &\coloneqq \lambda(x \colon \grtype).\ \mathsf{e} \wedge \neg x \wedge \subst{\sigma_{\rightarrow w}}{\big(\forall_{0}\dotsc\forall_{w-1}.\ \varphi_<^{w} \rightarrow x\big)} \\
                     & \qquad \qquad \;\;\;\; \vee \mathsf{e} \wedge x \wedge \subst{\sigma_{\rightarrow w}}{\big(\exists_0\dotsc\exists_{w-1}.\ \varphi_<^{w} \wedge \neg x\big)} \\
\exists^{(w)} x.\varphi &\coloneqq \big(\mu (F\colon \grtype \to \grtype).\, \lambda (x\colon\grtype).\, \varphi \vee F(\nxt^{(w)}(x))\big)\ \myfalse \\
\forall^{(w)} x.\varphi &\coloneqq \big(\nu (F\colon \grtype \to \grtype).\, \lambda (x\colon\grtype).\, \varphi \wedge F(\nxt^{(w)}(x))\big)\ \myfalse
\end{align*} 
The functioning of these formulas is explained in \cite{DBLP:journals/corr/abs-2209-10311}; the idea
is that they lift the total order $<$ to tuples and sets, and then use iteration through all good sets to model
existential, resp.\ universal quantification. The restriction to good sets stems from the
assumption on $\mathsf{e}$ and
the fact that the empty set (defined by $\myfalse$) is good, and that the incrementation pattern
generates good sets from good sets.
\end{proof}

\begin{restatable}{lemma}{quantifiergeneral}
\label{lem:quantifier-general}
Let $\varphi$ be tail recursive and fixpoint closed, let $\Transsys$ be an LTS and let $\eta$ be good. Let $k > 2$. Then there are functions $\exists^{w,k} x.\, \varphi$ and $\forall^{w,k} x.\, \varphi$ such that
$\overline{\state} \in \sem{\exists^{w,k} x.\, \varphi}{\Transsys}{\eta}$ iff there is $f \in \semm{\tau^\circ_{w,k}}{\Transsys}$ such that $\overline{\state} \in \sem{\varphi}{\Transsys}{\eta[x\mapsto f]}$, and $\overline{\state} \in \sem{\forall^{w,k} x.\, \varphi}{\Transsys}{\eta}$ iff $\overline{\state} \in \sem{\varphi}{\Transsys}{\eta[x\mapsto f]}$ for all $f \in \semm{\tau^\circ_{w,k}}{\Transsys}$.
 \end{restatable}

\begin{proof}
Again, we sketch the construction. Details can be found in \cite{DBLP:journals/corr/abs-2209-10311}. 
Consider the following formulas for $k \geq 1$, where $\overline{x} = x_1,\dotsc,x_w$ and 
$\overline{y} = y_1,\dotsc,y_w$:
\begin{align*}
\varphi_<^{w,k-1}(x_0,\dotsc,x_{w-1},y_0,\dotsc,y_{w-1}) &\coloneqq \bigvee_{i = 0}^{w-1} \varphi_<^{(w),k-1}(x_i, y_i)
                                                           \wedge \bigwedge_{j = 0}^{i-1} \neg \varphi_<^{(w),k-1}(x_j,y_j) \\
\varphi_<^{(w),k}(x,y)  &\coloneqq \exists^{w,k-1} \overline{x}.\ y(\overline{x}) \wedge \neg x(\overline{x}) \\
                                    &\; \wedge \forall^{w,k-1} \overline{y}.\ \varphi_<^{w,l}(\overline{y},\overline{x}) \rightarrow
																						    x(\overline{y}) \rightarrow y(\overline{y}) \\
\myfalse_{(w),k} & \coloneqq \lambda (\overline{x}\colon\tau_{w,k-1}).\, \myfalse \\
\nxt^{w,k}(x) &\coloneqq \lambda (x\colon\tau_{w,k}).\  \lambda (\overline{x}\colon\tau_{w,k-1}).\  \\
& \big( \neg x(\overline{x}) \wedge \forall^{w,k} \overline{y}.\ 
\varphi_<^{w,k}(\overline{y},\overline{x}) \rightarrow x(\overline{y})\big) \\
& \vee \big(  x(\overline{x}) \wedge \exists^{w,k} \overline{y}.\ 
\varphi_<^{w,k}(\overline{y},\overline{x}) \wedge \neg x(\overline{y})\big) \\
\exists^{w,k}(x).\ \varphi &\coloneqq \big( \mu (F\colon \tau_{w,k} \to \tau_{w,k}).\ \lambda (x\colon\tau_{w,k}).\, \varphi \vee F(\nxt^{w,k} x)\big)\, \myfalse_{(w),k} \\
\forall^{w,k}(x).\ \varphi &\coloneqq \big( \nu (F\colon \tau_{w,k} \to \tau_{w,k}).\ \lambda (x\colon\tau_{w,k}).\, \varphi \wedge F(\nxt^{w,k} x)\big)\, \myfalse_{(w),k} 
\end{align*}
where $\varphi_<^{(w),k-1} =  \varphi_<^{(w)}$ in case $k = 1$ and $\exists^{w,k-1} x. \varphi = \exists^{(w)} x. \varphi$ if $k=1$. 

These formulas follow the same pattern as those for quantification over sets, but instead they quantifiy
over $\semm{\tau^\circ_{w,k}}{\Transsys}$. The construction pattern is the same: lift $<$ to tuples and then
sets, and then use iteration instead of quantification. Goodness of the functions quantified over is 
due to the fact that $\nxt^{w,k}(x)$ is agnostic to individual states.
\end{proof} 

Hence, we extend $\SYNTRANS$ as follows:
\begin{align*}
\syntrans{\exists(X\colon \tau'_{w,k}).\ \varphi} &\coloneqq \exists^{w,k}(x).\ \syntrans{\varphi} \text{ if } k \geq 2 \\
\syntrans{\forall(X\colon \tau'_{w,k}).\ \varphi} &\coloneqq \forall^{w,k}(x).\ \syntrans{\varphi} \text{ if } k \geq 2 \\
\syntrans{X(X_0,\dotsc,X_{w-1})} &\coloneqq x(x_0,\dotsc,x_{w-1}) \text{ if } (X \colon \tau'_{w,k}) \text{ and } k \geq 2 %\\
\end{align*}

\begin{lemma}
\label{lem:lower-higher}
Assume that $\syntrans{\varphi}$ is tail-recursive, has no free fixpoint variables,
 and satisfies the invariant \eqref{eq:invariant}.
Then the translations of $\exists(X\colon \tau'_{w,k}).\ \varphi, \forall(X\colon \tau'_{w,k}).\ \varphi$ and $X(X_0,\dotsc,X_{w-1})$ %, \neg X(X_0,\dotsc,X_{w-1})$
for $(X \colon \tau'_{w,k}) \text{ with } k \geq 2$
are tail-recursive, have no free fixpoint variables, and satisfy the invariant \eqref{eq:invariant}, too.
\end{lemma}
The proof is by Lemmas~\ref{lem:quantifier-sets} and \ref{lem:quantifier-general}. % and \ref{lem:quant-tr}.

\paragraph*{Partial Fixpoints.}
A partial fixpoint of a function $f$ can be characterised as follows: start with the empty set,
and then apply $f$ a number of times, yielding $f^n(\emptyset)$ for some $n$. The fixpoint has been reached when 
$f^n(\emptyset) = f^{n+1}(\emptyset)$.

To encode this, we define a formula $\mathit{partial}(\varphi, X, Y_{i_0},\dotsc, Y_{i_{w-1}}, Z_{j_{0}},\dotsc,Z_{j_{w-1}})$,
for $X$ of type $\tau'_{w,i+1}$, and consequently $Y_{i_0},\dotsc, Y_{i_{w-1}}, Z_{j_0},\dotsc,Z_{j_{w-1}}$
of type $\tau'_{w,i}$, as follows for $i = 1$:
\begin{multline*}
\mathit{partial}(\varphi, X, Y_{i_0},\dotsc, Y_{i_{w-1}}, Z_{j_{0}},\dotsc,Z_{j_{w-1}}) = \\
\Big(\mu (F\colon \tau_{w,0} \to \tau_{w,0}).\ \lambda (x\colon \tau_{w,0}).\ F(\syntrans{\varphi}) \vee \\
\big(\{\sigma'\} x \wedge \forall_0, \dotsc, \forall_{w-1}.\ \big(x \leftrightarrow \syntrans{\varphi}\big)\big)\Big)\ \myfalse
\end{multline*}
Here, $\sigma'(i') = j'$ if $i'=i_k$ and $j'=j_k$, and otherwise $\sigma'(i') = i'$.
The intuition of the least fixpoint in this formula is that iteration starts with the empty set for $x$, and then 
either $\syntrans{\varphi}$ replaces $x$ in the next iteration (note that $\syntrans{\varphi}$ contains $x$ as
a free variable, so this amounts to unfolding the partial fixpoint once), or the current 
iteration of $x$ is assumed to be stable, in which case $x$ has to contain
$z_{j_0},\dotsc,z_{j_{w-1}}$, and survive a stability test: for \emph{all} elements $(z_{j_0},\dotsc,z_{j_{w-1}})$, 
$x$ contains the tuple $(z_{j_0},\dotsc,z_{j_{w-1}})$ iff $\varphi(x)$ does.\footnote{Due to the higher dimensionality, the tuples in question have arity $D$.} The latter
obviously only works if $x$ is already a fixpoint of $\varphi$. Hence,
a nontrivial least fixpoint in $\mathit{partial}(\varphi, X, Y_{i_0},\dotsc, Y_{i_{w-1}}, Z_{j_{0}},\dotsc,Z_{j_{w-1}})$
is reached iff such a partial fixpoint exists (and then the least fixpoint
agrees with the partial fixpoint modulo $\mathit{tptr}$), and otherwise the least fixpoint
is the empty set.

The above formula is defined as follows for $i > 1$:
\begin{multline*}
\mathit{partial}(\varphi, X, Y_{i_0},\dotsc, Y_{i_{w-1}}, Z_{j_0},\dotsc,Z_{j_{w-1}}) = \\
\mu (F\colon \tau_{w,i-1} \to \tau_{w,i-1}).\ \lambda (x\colon \tau_{w,i-1})\ F(\lambda (y_{i_0},\dotsc,y_{i_{w-1}} \colon \tau_{w,i-2}).\ \syntrans{\varphi}) \vee \\
\Big(x(z_{j_0},\dotsc,z_{j_{w-1}}) \wedge \forall^{(w,k)} z_{j_0}, \dotsc, \forall^{(w,k)} z_{j_{w-1}}.\\  x(z_{j_0},\dotsc,z_{j_{w-1}}) \leftrightarrow \syntrans{\varphi}(z_{j_0},\dotsc,z_{j_{w-1}})\Big)\ \myfalse_{\tau'_{w,i}}.
\end{multline*}
This is the
natural generalisation of the same construction for $i = 1$, using higher-order quantification instead of universal reachability.

We add the final case to $\SYNTRANS$, using both of the subcases above.
\begin{align*}
\syntrans{\big(\pfp(X,Y_{i_0},\dotsc,Y_{i_{w-1}}).\ \varphi\big)(Z_{j_0},\dotsc,Z_{j_{w-1}})} &\coloneqq 
\mathit{partial}(\varphi, X, Y_{i_0},\dotsc, Y_{i_{w-1}}, Z_{j_{0}},\dotsc,Z_{j_{w-1}}) % \\
\end{align*}

\begin{example}
\label{ex:trans-pfp}
Consider the formula $\varphi_{\mathsf{nonuniv}}$ from Ex.~\ref{ex:pfp}. We illustrate how its
first part, i.e.\ the formula $\exists Z.\ (\PFP(X,Y)\  \mathsf{step})(Z)$ is translated. For the sake of simplicity,
we pretend that $w = 3$. This works since no higher-order quantification happens in this formula.

First, the translation of $\mathsf{step}$ is as follows, where we assume that variable
$Y$ is encoded at index $1$, and variable $Y'$ is encoded at index $2$, and $\sigma_i$
moves the variable at index $i$ into the right position to be checked for set membership.
\begin{align*} 
\syntrans{\mathsf{step}(X, Y)} =\ & \quad \big(i_1 \wedge \forall_2 .\ \neg \{\sigma_2\} x\big) \\
                      & \vee \big( \{\sigma_1\} x \wedge \exists_2.\ \{\sigma_2\} x \wedge \forall_2. (\{\sigma_2\} x \rightarrow \neg f_2)\big)  \\
											& \vee \big(\exists_2.\ \mudiam{a_2}{} \varphi_\sim^{1,2} \wedge \exists_2.\ \{\sigma_2\}x \wedge f_2\big). 
\end{align*}
Now assume that the index of $z$ is $0$.
The translation of $\exists Z.\ (\PFP(X,Y)\  \mathsf{step})(Z)$ is then
\[
\exists_0.\ 
\Big(\mu F.\ \lambda x.\ F(\syntrans{\mathsf{step}}) \vee 
\big(\{\sigma_0\} x \wedge \forall_0 \forall_1 \forall_{2}.\  \big(x \leftrightarrow \syntrans{\mathsf{step}}\big)\big)\Big)\ \myfalse
\]
Here, the type of $F$ is $\tau_{w,0} \to \tau_{w,0}$ and that of $x$ is $\tau_{w,0}$.
\end{example}

\begin{lemma}
\label{lem:trans-partial}
Assume that $\syntrans{\varphi}$ is tail-recursive, has no free fixpoint variables,
 and satisfies the invariant \eqref{eq:invariant}.
Then so do the translations of $\pfp(X,Y_{0},\dotsc,Y_{w-1})$ for all types.
\end{lemma}
The proof is a straightforward application of invariant \eqref{eq:invariant} to the claim of the lemma, using
the above characterisation of partial fixpoints.

\subsection{The Capturing Result}

\begin{theorem}
\label{thm:capture}
Tail-recursive $\phfl{k+1}{}$ captures \bisim{k-\ExpSpace} for all $k \geq 0$.
\end{theorem}
\begin{proof}
The upper bound is shown in Thm.~\ref{thm:upper}. For the lower bound, it suffices, due to Prop.~\ref{prop:IV-gen}, 
to show that for any bisimulation-invariant query defined by an \hopfp{k+1} formula there is an equivalent,
tail-recursive 
formula in $\phfl{k+1}{}$. Let $\varphi$ be a formula 
defining such a query $Q$ of width $d \geq 1$, and w.l.o.g.\ $\varphi$ contains only the types $\tau_{w,0},\dotsc,\tau_{w,k}$ with $w \geq d$. Let $D = 2w + 2 + d$.

Let $\varphi' = \syntrans{\varphi}$ and let $\sigma$
be defined via $\sigma(2w+2+i) = i$ for $0 \leq i < w$, and $\sigma(j) = j$ for $j < 2w+2$. Note that $\sigma$ simply copies the first $d$ many
states in a tuple to the positions $2w+2,\dotsc,2w+2+d-1$ where they serve to retain goodness (cf.\ Def.~\ref{def:good}).
Then
$\psi = \big(\lambda (e\colon\grtype).\ \varphi'\big)\ \subst{\sigma}{\mytrue}$ defines the 
query $\{(\Transsys, (\state_0,\dotsc,\state_{D-1})) \mid (\Transsys, (\state_0,\dotsc,\state_{d-1})) \in Q\}$
by invariant \eqref{eq:invariant}.
Towards this, note that $\sem{\subst{\sigma}{\mytrue}}{\Transsys}{} = \States^{D-d} \times \{\state_0\} \times \dotsb \times \{\state_{d-1}\} =: M$.
Hence, $\sem{\psi}{\Transsys}{} = \sem{\varphi'}{\Transsys}{\eta}$ where $\eta = \emptyset[e\mapsto M]$.
\end{proof}

We observe that, technically, the resulting query has arity $D > d$. Since the result does not depend on the $D-d$ inputs beyond the first $d$,
this is morally a $d$-ary tuple (cf.\ also the discussion in \cite{Otto99}).

\section{Conclusion}
\label{sec:concl}

We have further completed the picture of the descriptive complexity of bisimulation-invariant
queries by adding characterisations for all levels of the exponential-space hierarchy, using fragments of PHFL. 
See Fig.~\ref{fig:bisimcomplexity} for a complete picture. Given the propensity of tail-recursive
programs for space-efficient evaluations, it is not surprising that the characterisation
of the bisimulation-invariant space hierarchy boils down to the tail-recursive fragment 
of PHFL. 

A natural follow-up target for bisimulation-invariant descriptive complexity are nondeterministic
classes, for example NP$/_\sim$. It is, of course, straightforward to obtain such a characterisation
as the image of existential second-order logic under the translation defined in this paper and 
\cite{DBLP:journals/corr/abs-2511-02596}, but clearly this does not yield a syntactically robust 
characterisation. Other characterisations are known \cite{DBLP:journals/iandc/Dawar98}, but
one that fits the framework of characterisations using fragments of PHFL still eludes discovery.
This is in stark contrast to the non-bisimulation-invariant setting, where the characterisation of NP was indeed
the first to be discovered \cite{Fagin74}. We remark that a syntactically robust characterisation of
bisimulation-invariant NP in terms of a fragment of \PHFL could be a contribution towards the resolution of the P-NP
problem, cf.\ \cite{DBLP:journals/corr/abs-2601-19641}.

\newpage
\bibliographystyle{eptcs}
\bibliography{literature}

\end{document}